\documentclass[11pt]{amsart}

\usepackage{graphicx}
\usepackage{amssymb}

\newtheorem{thm}{Theorem}[section]
\newtheorem{defn}[thm]{Definition}

\newtheorem{lem}[thm]{Lemma}
\newtheorem{remark}[thm]{Remark}

\newtheorem{ex}[thm]{Example}

\newcommand{\ds}{\displaystyle}

\begin{document}

\title{Characterization of Differentiable Copulas}


\author[Saikat Mukherjee]{S. Mukherjee$^\dagger$}
\address{Division of Science and Mathematics\\
University of Minnesota-Morris\\
Morris, MN 56267}
\email{smukherj@morris.umn.edu}
\thanks{$\dagger$ Corresponding author}

\author[Farhad Jafari]{F. Jafari}
\address{Department of Mathematics\\
University of Wyoming\\
Laramie, WY 82071}
\email{fjafari@uwyo.edu}

\author[Jong-Min Kim]{J. Kim}
\address{Division of Science and Mathematics\\
University of Minnesota-Morris\\
Morris, MN 56267}
\email{jongmink@morris.umn.edu}


\begin{abstract}
This paper proposes a new class of copulas which characterize the set of all twice continuously differentiable copulas. We show that our proposed new class of copulas is a new generalized copula family that include not only asymmetric copulas but also all smooth copula families available in the current literature.  Spearman's rho and Kendall's tau for our new Fourier copulas which are asymmetric are introduced. Furthermore, an approximation method is discussed in order to optimize Spearman's rho and the corresponding Kendall's tau.
\end{abstract}

\maketitle


\section{Introduction}

Recently, a study of dependence by using copulas has been getting more attention in the areas of finance, actuarial science, biomedical studies and engineering because a copula function does not require a normal distribution and independent, identical distribution assumptions. Furthermore, the invariance property of copula has been attractive in the finance area. But most copulas including Archimedean copula family are symmetric functions so that the copula model fitting for asymmetric data is not appropriate. Liebscher (\cite{Li08}) introduced two methods for the construction of asymmetric multivariate copulas. The first is connected with products of copulas. The second approach generalizes the Archimedean copulas. The resulting copulas are asymmetric but are little extension of the parametric families of copulas. This paper proposes a new generalized copula family which include asymmetric copulas in addition to all copula families available in the current literature.

We will characterize the set of all twice continuously differentiable copulas that may arise. We will start with some basic concepts of copulas in the next section. In Section 3 we characterize a class of differentiable copulas and state our main theorem. In Section 4 we will discuss some well-known copulas and introduce a new class of copulas (asymmetric in general) in support of our main result. Finally, in Section 5, we will study the dependence structure by calculating Spearman's rho and Kendall's tau and describe a method to maximize and minimize Spearman's rho using an approximation technique for a special class of copulas that arises from our construction.

\section{Definitions and Preliminaries}
In this section we recall some concepts and results that are necessary to understand a (bivariate) copula. Throughout this paper $\mathbb{I}$ denotes the unit interval $[0, 1]$. A copula is a multivariate distribution function defined on $\mathbb{I}^n$, with uniformly distributed marginals. In this paper, we focus on bivariate (two-dimensional, $n=2$) copulas.

\begin{defn}A bivariate copula is a function $C: \mathbb{I}^2\rightarrow \mathbb{I}$, which satisfies the following properties
\begin{enumerate}
\item[(P1)] $C(0,v)\;=\;C(u,0)\;=\;0,\qquad \forall u,v\in \mathbb{I}$
\item[(P2)] $C(1,u)\;=\;C(u,1)\;=\;u,\qquad \forall u\in \mathbb{I}$
\item[(P3)] $C(u_2,v_2)\;+\;C(u_1,v_1)\;-\;C(u_1,v_2)\;-\;C(u_2,v_1)\;\geq 0,\qquad \forall u_1, u_2, v_1, v_2\in \mathbb{I}$ with $u_1\leq u_2, v_1\leq v_2$.
\end{enumerate}
\end{defn}

The importance of copulas has been growing because of their applications in several fields of research. Their relevance primarily comes from the following theorem of Sklar (see \cite{Sk59}):
\begin{thm}\label{Sklar}
Let $X$ and $Y$ be two random variables with joint distribution function $H$ and marginal distribution functions $F$ and $G$, respectively. Then there exists a copula $C$ such that
\begin{center}
H(x,y) = C(F(x), G(y))
\end{center}
for all $x, y \in \mathbb{R}$. If $F$ and $G$ are continuous, then $C$ is unique. Otherwise, the copula $C$ is uniquely determined on $Ran(F)\times Ran(G)$. Conversely, if $C$ is a copula and $F$ and $G$ are distribution functions, then the function $H$ defined above is a joint distribution function with margins $F$ and $G$.
\end{thm}

Sklar's theorem clarifies the role that copulas play in the relationship between multivariate distribution functions and their univariate margins. A proof of this theorem can be found in \cite{ScSk83}.

\begin{defn}
Fr\'{e}chet lower and upper bounds for copulas are denoted by $C_L$ and $C_U$, respectively, and defined by
$$C_L(u, v) = \max\{u+v-1, 0\},$$
$$C_U(u, v) = \min\{u, v\},$$
for all $(u, v)\in \mathbb{I}^2.$
\end{defn}
Then it well-known that for any copula $C(u, v)$,
$$C_L(u, v) \leq C(u, v) \leq C_U(u,v),\qquad \forall (u, v)\in \mathbb{I}^2.$$

In this paper we will mainly concentrate on copulas which are twice continuously differentiable, i.e., $C(\cdot, \cdot) \in \mathcal{C}^2(\mathbb{I}^2)$. With this assumption and using property 3 (P3) of copulas, for all $u_1, u_2, v_1, v_2 \in \mathbb{I}$ with $u_1\leq u_2, v_1\leq v_2$, we have $C(u_2, v_2) - C(u_1, v_2) \geq C(u_2, v_1) - C(u_1, v_1)$. This implies $C(u_2, \cdot)-C(u_1, \cdot)$ is monotonically increasing in the second variable. Hence $\frac{\partial}{\partial v} \left[C(u_2, v)- C(u_1, v)\right] \geq 0$. Therefore $\frac{\partial}{\partial v} C(\cdot, v)$ is increasing in its first variable. Hence we deduce the following lemma:

\begin{lem}\label{copula_pde}
The following two statements are equivalent:
\begin{enumerate}
\item[i.] $C$ is a twice continuously differentiable copula.
\item[ii.] $C$  satisfies the following Dirichlet inequality problem:
\begin{equation}\label{co_pd_in}
\frac{\partial^2}{\partial u\partial v}C(u, v) \geq 0
\end{equation}
with boundary conditions:
$$C(u, 0) \;=\; C(0,v) \;=\; 0,$$
$$C(u, 1) \;=\; C(1,u) \;=\; u.$$
\end{enumerate}
\end{lem}

\section{Characterization of $\mathcal{C}^2$ copulas}

In this section we will start by solving the above problem stated in Lemma \ref{copula_pde} and then we will characterize all twice differentiable copulas.

Suppose $\gamma: \mathbb{I}^2\rightarrow \mathbb{R}$ is a continuous real-valued function. Then inequality (\ref{co_pd_in}) can be reformulated as follows,
$$\frac{\partial^2}{\partial u\partial v}C(u, v) = \gamma^2(u, v).$$
Integrating twice we get
$$C(u, v) = \ds\int_0^v\int_0^u \gamma^2(s, t)\; ds \;dt + H(u) + G(v),$$
where $H$ and $G$ are two arbitrarily real-valued functions of $u$ and $v$, respectively. Using boundary conditions $C(u, 0) = C(0, v) = 0$, we have $H(u) = -G(v) =\; $constant. Hence $C$ has the following form
\begin{equation}\label{copula1}
C(u, v) = \ds\int_0^v\int_0^u \gamma^2(s, t)\; ds \;dt.
\end{equation}
Now using the boundary condition $C(1, v) = v$, we have
$$\ds\int_0^v\int_0^1 \gamma^2(s, t)\; ds \;dt \;=\; v.$$
Differentiating both sides with respect to $v$, we have
\begin{equation}\label{copula2}
\int_0^1 \gamma^2(u, v)\; du \;=\; 1, \;\;\;\;\;\;\; \forall v\in \mathbb{I}.
\end{equation}
Similarly using the fourth boundary condition $C(u, 1) = u$, we have
\begin{equation}\label{copula3}
\int_0^1 \gamma^2(u, v)\; dv \;=\; 1, \;\;\;\;\;\;\; \forall u\in \mathbb{I}.
\end{equation}
This leads to the following theorem,

\begin{thm}\label{MJK}
Suppose $h: \mathbb{I}^2\rightarrow [-1, \infty)$ is a continuous real-valued function such that
\begin{eqnarray}
\int_0^1 h(u, v)\; dv \;=\; 0\;\;\;\;\;\; \forall u\in \mathbb{I},\label{copula_bd1}\\
\int_0^1 h(u, v)\; du \;=\; 0\;\;\;\;\;\; \forall v\in \mathbb{I},\label{copula_bd2}
\end{eqnarray}
then
\begin{equation}\label{copula}
C(u, v) = \ds\int_0^v\int_0^u 1 + h(s, t)\; ds \;dt.
\end{equation}
is a copula. Furthermore, every twice continuously differentiable copula is of the form given in (\ref{copula}).
\end{thm}
\begin{proof}
If we substitute $1 + h(u, v)$ for $\gamma^2(u, v)$ in Equations (\ref{copula1}), (\ref{copula2}) and (\ref{copula3}), it is easy to verify that every twice continuously differentiable copula is given by (\ref{copula}).

To prove the other direction, we have from (\ref{copula}), $\ds\frac{\partial^2C}{\partial u\partial v}\;=\; 1 + h(u, v)$, which is continuous and non-negative on $\mathbb{I}^2$. It is easy to check $C(u, 0) = C(0, v) = 0$ for all $(u, v)\in \mathbb{I}^2$. We also have,
\begin{eqnarray}
C(1, v) &=& \ds\int_0^v\int_0^1 1 + h(s, t)\; ds \;dt\nonumber\\
 &=& \ds\int_0^v 1 \; dt\; \nonumber\\
  &=& v. \nonumber
\end{eqnarray}
Similarly, we can show that $C(u, 1) = u$. Hence by Lemma \ref{copula_pde} $C$ is a copula.
\end{proof}

\begin{remark}
One importance of Theorem \ref{MJK} is the fact that in general there is no assumption of symmetry on $h(u, v)$ and hence on $C(u, v)$. With the help of above theorem, in the next section we will construct a class of examples of non-symmetric copulas.
\end{remark}

\section{Examples}

It is quite easy to verify Theorem \ref{MJK} for well-known $\mathcal{C}^2$ copulas by constructing the corresponding $h$. In this section we will begin with showing those renowned examples and later we will show how to construct other copulas.

\subsection{Archimedean Copulas}
Archimedean copula is a very interesting class of copulas, whose investigation arose in the context of associative functions and probabilistic metric spaces (see \cite{ScSk83}) and today has also many applications in the statistical context (see \cite{Ne99}). Moreover, Archimedean copulas are widely used in applications, especially in finance, insurance and actuarial science, due to their simple form and nice properties.

\begin{defn}
Let $\varphi: \mathbb{I}\rightarrow [0, \infty]$ be a continuous, decreasing function such that $\varphi(1) = 0$. The \textit{pseudo-inverse} of $\varphi$ is the function denoted $\varphi^{[-1]}$ with domain $[0, \infty]$, range $\mathbb{I}$ and defined by
\begin{equation}
\varphi^{[-1]}(t) = \ds \left\{\begin{array}{lcr}
 \varphi^{-1}(t) &\text{if}\quad 0\leq t \leq \varphi(0),\\
 0 & \text{if}\quad \varphi(0)\leq t \leq \infty\\
\end{array}\right.\nonumber\end{equation}
\end{defn}
\begin{defn}
Let $\varphi: \mathbb{I}\rightarrow [0, \infty]$ be a continuous, convex, strictly decreasing function such that $\varphi(1) = 0$. Then a copula $C$ is called \textit{Archimedean} if it can be written as
$$C(u, v) = \varphi^{[-1]}\left(\varphi(u) + \varphi(v)\right), \;\;\;\;\forall (u, v)\in \mathbb{I}^2.$$
$\varphi$ is called an additive generator of $C$.
\end{defn}

Now assuming $C\in \mathcal{C}^2(\mathbb{I}^2)$, let us define $$h(u, v) =  {\varphi^{[-1]}}^{''}\left(\varphi(u) + \varphi(v)\right) \varphi'(u) \varphi'(v) - 1.$$ Then it is evident that $h$ is continuous. Also one can easily show that $C(u, v) = \ds\int_0^u\int_0^v 1 + h(s, t)\;dt\;ds$ and $h$ satisfies both (\ref{copula_bd1}) and (\ref{copula_bd2}).

\begin{ex}[Frank Copulas]
The Frank copula is an Archimedean copula given by
$$C(u, v) = -\ds \frac{1}{\theta} \ln{\left\{1+\frac{(e^{-\theta u} -1)(e^{-\theta v} -1)}{e^{-\theta} -1}\right\}},$$
where $\theta \in \mathbb{R}\setminus \{0\}$ is a parameter. Its generator is given by
$$\varphi_{\theta}(x) = \ds -\ln{\left\{\frac{e^{-\theta x} -1}{e^{-\theta} -1}\right\}}.$$
Then $$\ds {\varphi_{\theta}^{[-1]}}^{''}(x) = \frac{1}{\theta}\frac{e^{\theta + x} (e^\theta-1)}{(1 - e^\theta + e^{\theta + x})^2}, \qquad \varphi_\theta '(x)=\frac{\theta e^{-\theta x}}{e^{-\theta x}-1}.$$
Therefore we have
$$h(u, v) = \ds \frac{\theta e^{\theta (1 + u + v)} (e^\theta -1) }{(e^\theta-e^{\theta(1+u)}-e^{\theta(1+v)}+e^{\theta(u+v)})^2}.$$
\end{ex}
One can easily verify that $h$ satisfies Theorem \ref{MJK}.

\subsection{FGM Copulas}
The Farlie-Gumbel-Morgenstern (FGM) copula takes the form
$$C(u, v) = uv(1 + \theta (1-u)(1-v)),$$
where $\theta \in [-1, 1]$ is a parameter. The FGM copula was first proposed by Morgenstern (1956). The FGM copula is a perturbation of the product copula; if the dependence parameter $\theta$ equals zero, then the FGM copula collapses to independence. This is attractive primarily because of its simple analytical form. FGM distributions have been widely used in modeling, for tests of association, and in studying the efficiency of nonparametric procedures. However, it is restrictive because this copula is only useful when dependence between the two marginals is modest in magnitude.

Let $h(u, v) = \theta (1-2u)(1-2v)$. It can be easily verified that $h$ agrees with Theorem \ref{MJK}, and it generates the FGM copulas.

\subsection{Fourier Copulas}
The following lemma will introduce a new class of $\mathcal{C}^2$ copulas.
\begin{lem}\label{MJK_lem1}
Suppose $(a_n), (b_n), (c_n), (d_n) \in \ell^1$, the space of sequences whose series is absolutely convergent, are sequences of real numbers. Also suppose $h$ is a real-valued function on $\mathbb{I}^2$ defined by
$$h(u, v) = \ds\sum_{n=1}^\infty \left(a_n\cos{(2\pi nu)} + b_n\sin{(2\pi nu)}\right) \sum_{m=1}^\infty \left(c_m\cos{(2\pi mv)} + d_m\sin{(2\pi mv)}\right),$$
then
$$\ds\int_0^1h(u,v)\; du= 0,\;\;\;\;\;\; \forall v\in{\mathbb{I}},$$
$$\ds\int_0^1h(u,v)\; dv= 0,\;\;\;\;\;\; \forall u\in{\mathbb{I}}.$$
Furthermore, if $\ds\sum_{n, m=1} ^\infty \sqrt{a_n^2 + b_n^2} \sqrt{c_m^2+d_m^2} \leq 1$, then $h(u, v)\geq -1$ for all $(u, v) \in \mathbb{I}^2$.
\end{lem}

\begin{proof}
Notice that we can rewrite $h(u, v)$ as,
\begin{eqnarray}
h(u, v) &=& \ds\sum_{n, m} a_n c_m\cos{(2\pi nu)}\cos{(2\pi mv)} + \sum_{n, m} a_n d_m\cos{(2\pi nu)}\sin{(2\pi mv)}\nonumber\\
 & & + \sum_{n, m} b_n c_m\sin{(2\pi nu)}\cos{(2\pi mv)} + \sum_{n, m} b_n d_m\sin{(2\pi nu)}\sin{(2\pi mv)}.\nonumber
\end{eqnarray}
The first conclusion follows from the fact that the sequences $(a_n), (b_n), (c_n), (d_n) \in \ell^1$ and $\ds\int_0^1 \sin{(2\pi nx)}\;dx = \ds\int_0^1 \cos{(2\pi nx)}\;dx = 0$, $\forall n\in \mathbb{N}$.

Now to prove $h(u, v) \geq -1$ for all $(u, v) \in \mathbb{I}^2$, first notice that $\ds \sum_n \sqrt{a_n^2+b_n^2} \leq \sum_n \left(|a_n| + |b_n|\right) < \infty$. Also,
$$-\sqrt{a_n^2+b_n^2} \leq a_n\cos{(2\pi n u)} + b_n \sin{(2\pi n u)} \leq \sqrt{a_n^2+b_n^2},$$
for all $u\in \mathbb{I},\; n\in \mathbb{N}$. Hence we have
$$\ds -\sum_{n, m}\sqrt{a_n^2+b_n^2}\sqrt{c_m^2+d_m^2} \leq h(u, v) \leq \sum_{n, m}\sqrt{a_n^2+b_n^2}\sqrt{c_m^2+d_m^2}.$$
Therefore the additional hypothesis on $h$ guarantees that $h(u, v) \geq -1$.
\end{proof}

It is evident that $h(\cdot, \cdot)$ is continuous and therefore by Theorem \ref{MJK}, $C$, defined by
$$C(u, v) = \ds\int_0^v\int_0^u 1 + h(s, t)\; ds \;dt,$$
is a copula, where $h$ is of the form given in Lemma \ref{MJK_lem1}.

Noting that
$$ \sum_{n=1}^\infty a_n \cos (2\pi n u) + b_n \sin(2 \pi nu) = \sum_{n \in \mathbb{Z}\setminus \{0\}} \gamma_n e^{2\pi i nu}, $$
where $ a_n = \gamma_n + \gamma_{-n} $ and $ b_n = i (\gamma_n - \gamma_{-n}) $, or equivalently,
$$ \gamma_n = \overline{\gamma_{-n} } = \frac{a_n - i b_n}{2}, $$ it follows that $ |\gamma_n| = \frac{\sqrt{a_n^2 + b_n^2}}{2} $.  Similarly, if $ \delta_m = \overline{\delta_{-m}} = \frac{c_m - i d_m}{2} $,
$$ h(u,v) = \sum_{n,m \in \mathbb{Z}\setminus\{0\}} \gamma_n \delta_m \exp(2\pi i (nu + mv)), $$ with sequences $ (\gamma_n), (\delta_m) \in  \ell^1 $, and $ ||\gamma_n||||\delta_m|| \leq \frac{1}{4} $, then $ h $ will satisfy the conclusions of
Lemma \ref{MJK_lem1}, and will generate a copula by eq. (\ref{copula}). This restatement of Lemma \ref{MJK_lem1} clearly indicates that $ h $ is obtained from products of functions in the disc algebra with vanishing zero-th moment and where the product of the $ \ell^1 $ norms of the coefficients are $ \leq \frac{1}{4} $.  More precisely,

\begin{thm}\label{generator}
If $ h $ is a function on the $ 2 $-torus arising from the product of functions on the unit circle, each section of $ h $ has a vanishing zero-th moment, and the product of $ \ell^1 $ norms of Fourier coefficients of components of $ h $ is $ \leq \frac{1}{4} $, then $ h$ generates a $ C^2 $-copula.
\end{thm}

This theorem provides a large class of examples from which one may construct copulas with optimal properties ($\rho $ and $ \tau $).

\begin{ex}\label{four_cop_ex}
Let $b_1 = c_1 = 1, b_n = c_n = 0,\; \forall n\neq 1; a_n= d_n = 0\; \forall n$, then
$$h(u,v) = \sin{(2\pi u)}\cos{(2\pi v)}.$$
Define $C$ as follows
\begin{eqnarray}
C(u, v) &=& \ds\int_0^v\int_0^u 1 + \sin{(2\pi s)}\cos{(2\pi t)}\; ds \;dt\nonumber\\
&=& \ds\int_0^v u + \frac{1}{2\pi}\{1-\cos{(2\pi u)}\}\cos{(2\pi t)}\; ds \;dt\nonumber\\
&=& \ds uv + \frac{1}{4\pi^2}\{1-\cos{(2\pi u)}\}\sin{(2\pi v)}\nonumber
\end{eqnarray}
Then it is easy to verify that $C$ forms a copula and more importantly it is not symmetric. A contour plot of $C$ is given in Figure \ref{fig:fourier-cop}.
\begin{figure}[ht]
\centering
\includegraphics[width=0.6\textwidth]{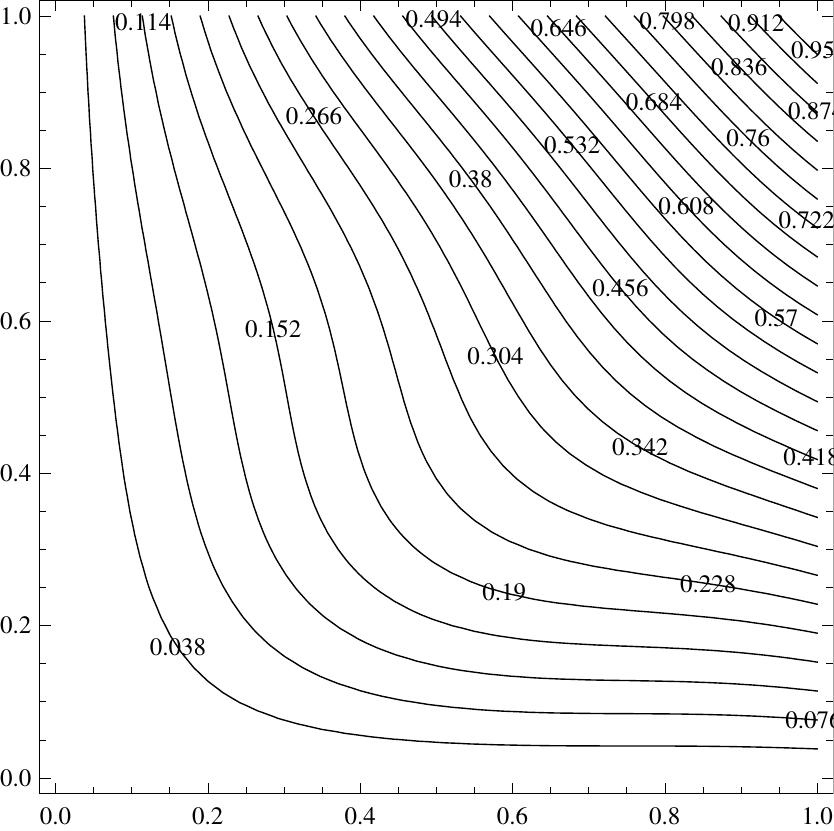}
\caption{Contour plot  of a Fourier copula in Ex. \ref{four_cop_ex}}
\label{fig:fourier-cop}
\end{figure}

\end{ex}

\section{Spearman rho and Kendall tau}

The two most commonly used nonparametric measures of association for two random variables are Spearman rho ($\rho$) and Kendall tau ($\tau$). In general they measure different aspects of the dependence structures and hence for many joint distributions these two measures have different values.
\begin{defn}
Suppose X and Y are two random variables with marginal distribution functions F and G, respectively. Then \textit{Spearman} $\rho$ is the ordinary (Pearson) correlation coefficient of the transformed random variables F(X) and G(Y), while \textit{Kendall} $\tau$ is the difference between the probability of concordance $Pr[(X1 - X2)(Y1 - Y2)>0]$ and the probability of discordance $Pr[(X1 - X2)(Y1 - Y2)<0]$ for two independent pairs (X1, Y1) and (X2, Y2) of observations drawn from the distribution.
\end{defn}

In terms of dependence properties, Spearman $\rho$ is a measure of average quadrant dependence, while Kendall $\tau$ is a measure of average likelihood ratio dependence (see \cite{Ne99} for details). If $X$ and $Y$ are two continuous random variables with copula $C$, then Kendall $\tau$ and Spearman $\rho$ of $X$ and $Y$ are given by,
\begin{equation}\label{tau}
\tau = \ds 4\int\int_{\mathbb{I}^2} C(u, v) \;dC(u, v) - 1
\end{equation}
\begin{equation}\label{rho}
\rho = \ds 12\int\int_{\mathbb{I}^2} C(u, v) \;du\; dv - 3
\end{equation}

\subsection{$\rho$, $\tau$ for Fourier copulas}
From Lemma (\ref{MJK_lem1}), if we define
$$h(u, v) = \ds\sum_n \left(a_n\cos{(2\pi nu)} + b_n\sin{(2\pi nu)}\right) \sum_m\left(c_m\cos{(2\pi mv)} + d_m\sin{(2\pi mv)}\right),$$
where $\ds\sum_{n, m} \sqrt{a_n^2 + b_n^2} \sqrt{c_m^2+d_m^2} \leq 1$, then the Fourier copulas are given by
\begin{eqnarray}
C(u, v) &=& \int_0^v\int_0^u 1 + h(s, t)\; ds \;dt \nonumber\\
&=& uv + \sum_{n, m} \frac{1}{4\pi^2nm}\left[a_n\sin{(2\pi nu)} + b_n\{1-\cos{(2\pi nu)}\}\right] \nonumber\\
& & \qquad \left[c_m\sin{(2\pi mu)} + d_m\{1-\cos{(2\pi mv)}\}\right]\nonumber
\end{eqnarray}
Using (\ref{tau}) and (\ref{rho}), we have
$$\rho = \frac{3}{\pi^2} \sum_{n, m}\frac{b_nd_m}{nm},$$
$$\tau = \frac{2}{\pi^2} \sum_{n, m}\frac{b_nd_m}{nm}.$$
Now since $\ds -\sum_{n, m} \sqrt{a_n^2 + b_n^2} \sqrt{c_m^2+d_m^2} \leq \sum_{n, m}\frac{a_nc_m}{nm} \leq \sum_{n, m} \sqrt{a_n^2 + b_n^2} \sqrt{c_m^2+d_m^2}$, we can conclude that

$$-0.304 \approx -\frac{3}{\pi^2} \leq \rho \leq \frac{3}{\pi^2} \approx 0.304,$$
$$-0.203 \approx -\frac{2}{\pi^2} \leq \tau \leq \frac{2}{\pi^2} \approx 0.203.$$

\begin{remark}
\emph{Fourier copulas can be generalized by writing $h$ as follows,
$$h(s, t) = \sum_{n ,m \in \mathbb{Z}\setminus \{0\}} \alpha_{n, m}\; \exp(2\pi i (ns + mt)),$$
where $\ds \alpha_{n, m} = \overline{\alpha_{-n, -m}} \; \forall n, m \in \mathbb{Z}\setminus \{0\}$ and $\ds \sum_{n, m \in \mathbb{N}} |\alpha_{n, m}| + |\alpha_{-n, m}|$$ \leq \frac{1}{2}$. The latter condition here is to ensure that $ h $ will have range in $ [-1, \infty) $. Notice that for all $n, m \neq 0$, $\alpha_{-n,m}$ is equal to $\overline{\alpha_{n,-m}}$ and hence no additional conditions are necessary to assure $h$ to be real-valued.\\ This yields $\ds \rho = -\frac{3}{\pi^2} \sum_{n, m\in\mathbb{Z}\setminus \{0\}}\frac{\alpha_{n, m}}{nm},$ and $\ds \tau = -\frac{2}{\pi^2} \sum_{n, m\in\mathbb{Z}\setminus \{0\}}\frac{\alpha_{n, m}}{nm},$ and it can be shown again that $ -\frac{3}{\pi^2} \leq \rho \leq \frac{3}{\pi^2}$, and $ -\frac{2}{\pi^2} \leq \tau \leq \frac{2}{\pi^2}.$}
\end{remark}

\subsection{Optimizing $\rho$ for $\mathcal{C}^2$-copulas}
In this section we will optimize Spearman's rho using an approximation method for $\mathcal{C}^2$-copulas with $h$ of the form, $h(x, y) = \varphi(x)\psi(y)$, where $\varphi$ and $\psi$ both are continuous real-valued functions on $\mathbb{I}$. Notice that in this special case, $\rho$ can be simplified into the following form,
\begin{eqnarray}
\rho &=& 12\int\int_{\mathbb{I}^2} \left[\int_{t=0}^v \int_{s=0}^u 1 + \varphi(s) \psi(t) ds\;dt\right] du\;dv - 3\nonumber\\
&=& 12\int_0^1 \int_0^u \varphi(s) ds\;du\int_0^1 \int_0^v \psi(t) dt\;dv.\nonumber
\end{eqnarray}
This suggests that optimizing $\rho$ is equivalent to optimizing both $\int_0^1 \int_0^u \varphi(s) ds\;du$ and $\int_0^1 \int_0^v \psi(t) dt\;dv$.

Define $G(u):=\int_0^u \varphi(s) ds$ and $H(v):=\int_0^v \psi(t) dt$. Then for some positive $\alpha_1, \alpha_2, \beta_1, \beta_2$, the optimization problems become,

\noindent\begin{minipage}{.5\linewidth}
\begin{equation*}
\begin{aligned}
& {\text{max/min}}
& & I_1 = \int_0^1 G(u)\; du \\
& \text{subject to}
& & G(0) = G(1) = 0 \\
&&& -\alpha_1\leq G'(u) \leq \beta_1,
\end{aligned}
\end{equation*}
\end{minipage}%
\begin{minipage}{.5\linewidth}
\begin{equation*}
\begin{aligned}
& {\text{max/min}}
& & I_2 = \int_0^1 H(v)\; du \\
& \text{subject to}
& & H(0) = H(1) = 0 \\
&&& -\alpha_2\leq H'(v) \leq \beta_2.
\end{aligned}
\end{equation*}
\end{minipage}\\

Although it apparently looks like these optimization problems can be solved independently,  they are related by the fact that $G'(u) H'(v) = \varphi(u) \psi(v) \geq -1$. This implies $\min\{-\alpha_1\beta_2, -\alpha_2\beta_1\} \geq -1$. For the optimal possibility, we choose, $\beta_2=-(\alpha_1)^{-1}$ and $\alpha_2=-(\beta_1)^{-1}$. This is evident from the fact that both $I_1$ and $I_2$ can be positive or negative, $\rho_{\text{max}}$ will occur either if both $I_1$ and $I_2$ are maximum or if both are minimum and $\rho_{\text{min}}$ will occur if one of $I_1$ and $I_2$ is maximum and the other is minimum.

Geometrically, $I_1$ will be maximum if $G$ has the form as in Figure \ref{fig:Gmax} and will be minimum if $G$ has the form as in Figure \ref{fig:Gmin}.

\begin{figure}[htbp]
\begin{minipage}{0.5\linewidth}
\centering
\includegraphics[width=0.8\textwidth]{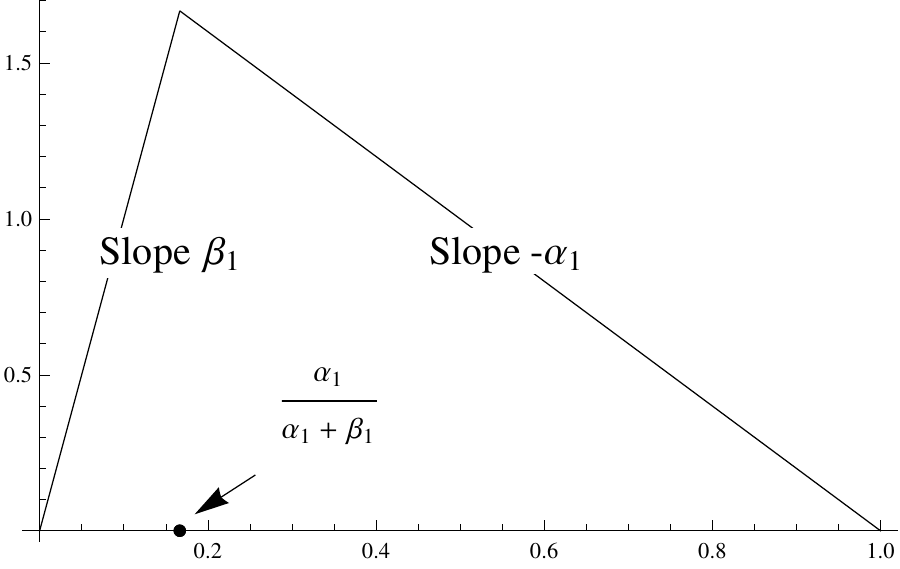}
\caption{$G$, Maximizing $I_1$}
\label{fig:Gmax}
\end{minipage}%
\begin{minipage}{0.5\linewidth}
\centering
\includegraphics[width=0.8\textwidth]{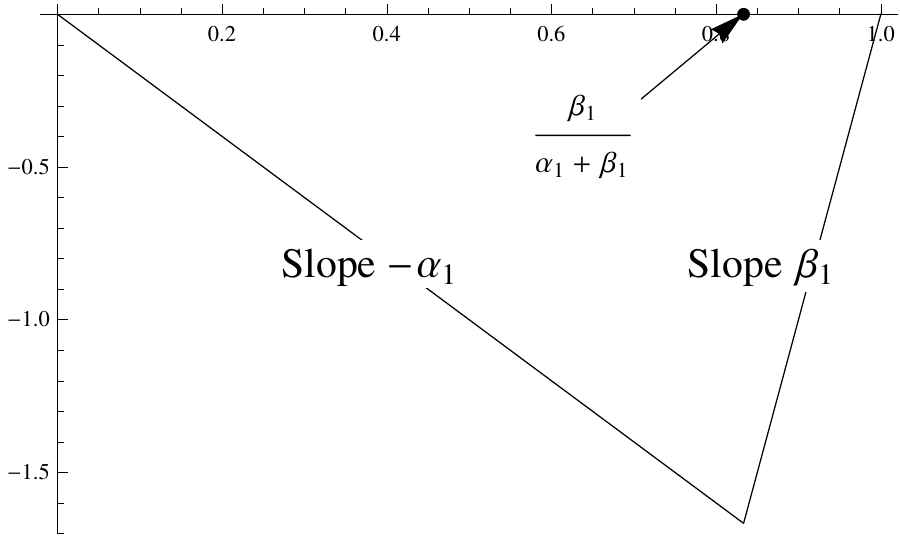}
\caption{$G$, Minimizing $I_2$}
\label{fig:Gmin}
\end{minipage}
\end{figure}

One can easily prove that, in order to optimize $I_1$, $\beta_1$ must be equal to $\alpha_1$. For convenience, now onwards we will write $\alpha$ for $\alpha_1$. This suggests that if $G(x) = GM(x) = -\alpha |x-0.5|+0.5\alpha$, or $G(x) = Gm(x) = \alpha |x-0.5|-0.5\alpha$ then $I_1$ will be maximum or minimum, respectively. But in either case, $G$ is not differentiable at $x=0.5$, and hence $\varphi$ is not continuous. To avoid this, we will approximate $G$ by a smooth function as follows: for arbitrarily small $\varepsilon >0$, define $$\widetilde{GM}(x) = -\widetilde{Gm}(x) = \frac{\alpha}{2}\left(\sqrt{1+4\varepsilon^2}-\sqrt{(1-2x)^2+4\varepsilon^2}\right).$$ It is worth noting that $\ds\sup_{x\in \mathbb{I}} \Big\{|\widetilde{GM}(x)-GM(x)|, |\widetilde{Gm}(x)-Gm(x)|\Big\} \rightarrow 0$ as $\varepsilon \rightarrow 0$ and $-\alpha \leq \widetilde{GM}'(x), \widetilde{Gm}'(x) \leq \alpha$.

\begin{figure}[htbp]
\begin{minipage}{0.5\linewidth}
\centering
\includegraphics[width=0.8\textwidth]{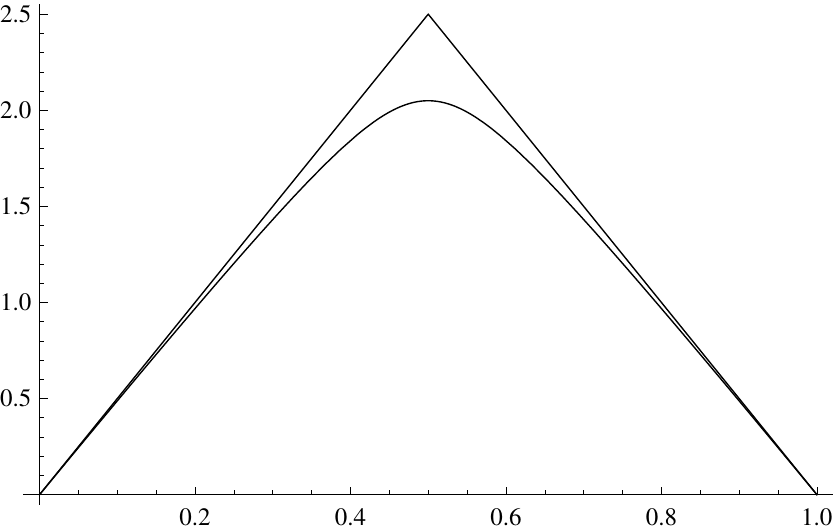}
\caption{$GM$, $\widetilde{GM}$ for $\alpha=5$, $\varepsilon = 0.1$}
\label{fig:Gapprox1}
\end{minipage}%
\begin{minipage}{0.5\linewidth}
\centering
\includegraphics[width=0.8\textwidth]{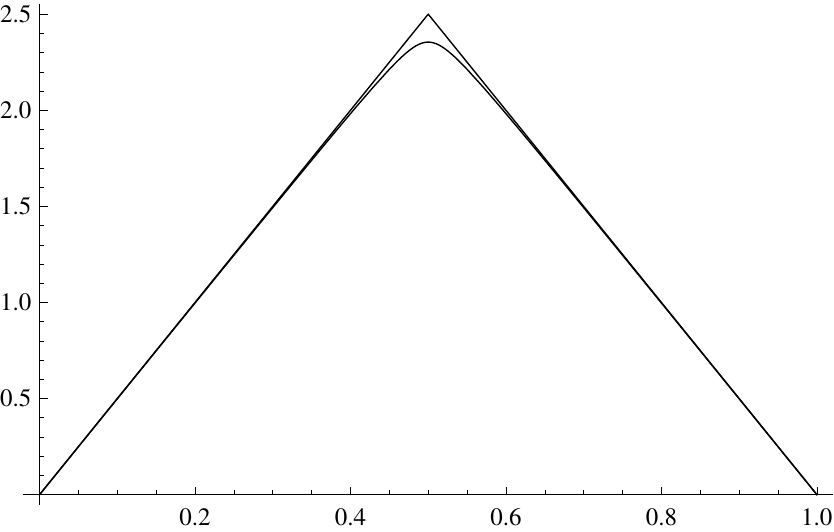}
\caption{$GM$, $\widetilde{GM}$ for $\alpha=5$, $\varepsilon = 0.03$}
\label{fig:Gapprox2}
\end{minipage}
\end{figure}

We can similarly optimize $I_2$ by approximating maximum and minimum of $H$ by the following functions $$\widetilde{HM}(x) = -\widetilde{Hm}(x) = \frac{1}{2\alpha}\left(\sqrt{1+4\varepsilon^2}-\sqrt{(1-2x)^2+4\varepsilon^2}\right).$$

Hence optimum values of $\rho$ will be obtained by approximating $h$ by the following functions,
$$h^\varepsilon_{\text{max}}(x, y) = \widetilde{GM}'(x)\widetilde{HM}'(y) = \widetilde{GM}'(x)\widetilde{HM}'(y) = \frac{(1-2x)(1-2y)}{\sqrt{(1-2x)^2+4\varepsilon^2}\sqrt{(1-2y)^2+4\varepsilon^2}},$$
$$h^\varepsilon_{\text{min}}(x, y) = \widetilde{GM}'(x)\widetilde{Hm}'(y) = \widetilde{Gm}'(x)\widetilde{HM}'(y) = -\frac{(1-2x)(1-2y)}{\sqrt{(1-2x)^2+4\varepsilon^2}\sqrt{(1-2y)^2+4\varepsilon^2}}.$$

Notice that each of $h^\varepsilon_{\text{max}}$ and $h^\varepsilon_{\text{min}}$ will generate a copula as it satisfies all the hypothesis of Theorem \ref{MJK}. Then the corresponding Spearman's rho and Kendall's tau are given by,
$$\rho^\varepsilon_{\text{max}} = \frac{3}{4} \left(\sqrt{1 + 4 \varepsilon^2} - 4 \varepsilon^2 \coth^{-1}(\sqrt{1 + 4 \varepsilon^2})\right)^2$$
$$\rho^\varepsilon_{\text{min}} = -\frac{3}{4} \left(\sqrt{1 + 4 \varepsilon^2} - 4 \varepsilon^2 \coth^{-1}(\sqrt{1 + 4 \varepsilon^2})\right)^2$$

$$\tau^\varepsilon_{\text{max}} = \frac{1}{2} \left[1 + 4 \varepsilon^2 + 4 \varepsilon^2 \left(\sqrt{1 + 4 \varepsilon^2} - 2 \varepsilon^2 \coth^{-1}(\sqrt{1 + 4 \varepsilon^2})\right) \ln\left(\frac{
     1 + 2 \varepsilon^2 - \sqrt{1 + 4 \varepsilon^2}}{2 \varepsilon^2}\right)\right]$$
$$\tau^\varepsilon_{\text{max}} = -\frac{1}{2} \left[1 + 4 \varepsilon^2 + 4 \varepsilon^2 \left(\sqrt{1 + 4 \varepsilon^2} - 2 \varepsilon^2 \coth^{-1}(\sqrt{1 + 4 \varepsilon^2})\right) \ln\left(\frac{
     1 + 2 \varepsilon^2 - \sqrt{1 + 4 \varepsilon^2}}{2 \varepsilon^2}\right)\right]$$

The optimal values of $\rho$ and corresponding $\tau$ will be obtained by letting $\varepsilon \rightarrow 0$. Table \ref{table:optization} shows how the values of $\rho$ approach the optimal values as $\varepsilon \rightarrow 0$ and it is clear that $-0.75 \leq \rho \leq 0.75$ and $-0.5 \leq \tau \leq 0.5$.

\begin{table}[t]
\centering
\begin{tabular}{|c|c|c|c|c|}
\hline
$\varepsilon $ & $\rho^\varepsilon_{\text{max}}$  &$\rho^\varepsilon_{\text{min}}$ & $\tau^\varepsilon_{\text{max}}$  &$\tau^\varepsilon_{\text{min}}$\\
\hline
\hline
1   & 0.0726437  & -0.0726437 &  0.0484292 & -0.0484292\\
\hline
0.1   & 0.644923  & -0.644923 & 0.429949  & -0.429949\\
\hline
0.01   & 0.747539  & -0.747539 & 0.498359  & -0.498359\\
\hline
0.001   & 0.749962  & -0.749962 &  0.499974 & -0.499974\\
\hline
0.0001   & 0.749999  & -0.749999 & 0.5  & -0.5\\
\hline
\end{tabular}
\caption{$\rho$ and $\tau$ values as $\varepsilon$ changes.}
\label{table:optization}
\end{table}

\section{Conclusion}
We proposed a new generalized copula family which include not only asymmetric copulas but also all copula families available in the current literature. Especially, the family of Fourier copulas we proposed is very useful copula family for analyzing asymmetric data such as financial return data or cancer data in Bioinformatics.  In our future study, we will extend our copula method to a multivariate case and then incorporate time varying component to our proposed method.

\bibliographystyle{amsalpha}

\end{document}